\documentclass{article}
\usepackage[english]{babel}
\usepackage[utf8]{inputenc}
\usepackage{spconf,ifpdf}
\usepackage{cite} 
\usepackage{url}
\usepackage[draft,bookmarks=false]{hyperref} 
\ifpdf
	\usepackage[pdftex]{graphicx}
	\graphicspath{{./figures/}}
\else
	\usepackage[dvips]{graphicx}
	\graphicspath{./figures/}
\fi
\usepackage{color}
\usepackage{pgf, tikz, pgfplots}

\usetikzlibrary{shapes, arrows, automata}
\usepackage{caption} 
\usepackage{subcaption} 
\usepackage{amsmath}
\usepackage{amsfonts, amssymb, amsthm}
\usepackage{mathrsfs}
\usepackage{upgreek}
\usepackage{algorithm,algpseudocode}
\usepackage{enumerate}
\usepackage{multirow}
\usepackage{rotating}
\usepackage{bm}
\usepackage{dsfont}
\usepackage{needspace}

\input{mySymbol.sty}

\newtheorem{lemma}{\hspace{0pt}\bf Lemma}

\newtheorem{theorem}{\hspace{0pt}\bf Theorem}

\newtheorem{definition}{\hspace{0pt}\bf Definition}

\newcommand{\QED}{\hfill\ensuremath{\blacksquare}}

\title{Convolutional Filtering on Sampled Manifolds}
{\name{Zhiyang Wang$^{\dagger}$ \qquad Luana Ruiz${^\star}$ \qquad Alejandro Ribeiro$^{\dagger}$}
\address{  $^{\dagger}$Department of Electrical and Systems Engineering, University of Pennsylvania, PA \\
${^\star}$ Simons-Berkeley Institute, CA \\ 
\thanks{Supported by NSF CCF 1717120, Theorinet Simons. } }
}

\begin{document}
\ninept
\maketitle
%
\begin{abstract}
The increasing availability of geometric data has motivated the need for information processing over non-Euclidean domains modeled as manifolds. 
The building block for information processing architectures with desirable theoretical properties such as invariance and stability is convolutional filtering. 
Manifold convolutional filters are defined from the manifold diffusion sequence, constructed by successive applications of the Laplace-Beltrami operator to manifold signals. However, the continuous manifold model can only be accessed by sampling discrete points and building an approximate graph model from the sampled manifold. Effective linear information processing on the manifold requires quantifying the error incurred when approximating manifold convolutions with graph convolutions. In this paper, we derive a non-asymptotic error bound for this approximation, showing that convolutional filtering on the sampled manifold converges to continuous manifold filtering. Our findings are further demonstrated empirically on a problem of navigation control. 
\end{abstract}
%


\section{Introduction} \label{sec:intro}

Applications involving structured data in non-Euclidean domains are attracting more and more attention these days, including but not limited to robot coordination \cite{li2020graph, wang2020learning}, resource allocation over wireless networks \cite{wang2022stable} and 3D shape analysis \cite{zeng20193d, devanne20143}. Graph convolutional filtering \cite{ortega2018graph, gama2019convolutional} and manifold convolutional filtering \cite{wang2022convolutional, masci2015geodesic} have become the top choices to process information over these non-Euclidean domains. This can be attributed to the fact that the convolution operation consists of local diffusions over the geometric structure, which capture features that are shared across the whole domain. Moreover, convolutional filtering provides the fundamental block for constructing deep learning architectures, which are powerful models \cite{bronstein2017geometric} to solve a wide range of geometric problems. 

Unlike graphs, which are discrete objects, manifolds are not directly accessible in general because they are infinite-dimensional continuous spaces \cite{zeng20193d, wang2022convolutional}. Instead, we typically have access to a set of points sampled from the manifold, or a \emph{sampled manifold}. In many manifold problems, a \emph{graph} capturing both the local and global connection structure is built from this sampled manifold to approximate the underlying continuous manifold \cite{shi2020point, wang2022stable} 
(herein, we use the terms \emph{sampled manifold} and \emph{graph} interchangeably). Convolutional filters constructed on graphs have been used to approximately process information over the manifold, which has been supported both experimentally and theoretically \cite{bronstein2017geometric, monti2017geometric, wang2022convolutional}. However, an explicit non-asymptotic theoretical result closing the gap between graph convolutional filtering and manifold convolutional filtering is still lacking.

\noindent\textbf{Contributions.} In this paper, we define manifold filtering as a convolutional operation based on diffusions that are defined as exponentials of the Laplace-Beltrami (LB) operator $\ccalL$ of the embedded manifold $\ccalM \subset \reals^\mathsf{N}$ \cite{wang2022convolutional}. Graphs are constructed by sampling a set of points uniformly over the manifold, with edge weights given by application of a Gaussian kernel. Denoting the graph Laplacian $\bbL_n^\epsilon$, we derive upper bounds for both the operator norm difference and the difference between the spectra of the graph Laplacian $\bbL_n^\epsilon$ and the LB operator $\ccalL$. Then, to deal with the infinite dimensional spectrum and the sharp variations in the high frequency domain of LB operator $\ccalL$, we import the definition of the Frequency Difference Threshold (FDT) filters (Definition \ref{def:alpha-filter}) \cite{wang2021stability} to separate the spectrum into finite eigenvalue partitions. Finally, we show that when applying FDT filters to graphs sampled from the manifold, graph filtering converges to the manifold filtering with a non-asymptotic approximation bound in the order of $n^{-1/(2d+8)}$ (Theorem \ref{thm:converge-MNN}).

\noindent\textbf{Related Works.} Graph convolutional filtering has been thoroughly discussed in previous works \cite{ortega2018graph, shuman2013emerging, gama2019convolutional}. A large number of works have studied the convergence of graph filters, such as \cite{ruiz2020graphon} and \cite{keriven2020convergence} which consider graphs that are dense or relatively sparse, and which are both sampled from a graphon model. Unlike the graphon, the manifold is more realistic and can be seen as a limit for graphs that have finite degree \cite{calder2019improved}. Convolutional filtering in manifolds has been discussed in \cite{wang2022convolutional, masci2015geodesic,monti2017geometric}, with different definitions depending on how the graphs are sampled from manifolds. The connection between graph filtering and manifold filtering has been discussed in \cite{wang2022convolutional}. In this paper, it is proved that manifold filtering can recover graph filtering by discretization in both the space and time domains; and that filtering on graphs sampled from the manifold converges to manifold filtering. However, there are no theoretical convergence rates. In the present paper, we extend upon the results of \cite{wang2022convolutional} by proving a non-asymptotic approximation bound for graph filters sampled from manifold filters.





\section{Preliminary definitions} \label{sec:prelim}

In order to relate graph filters and manifold filters, we start by reviewing the basic setup as well as the graph signal processing and manifold signal processing frameworks.

\subsection{Graph signal processing: graph filters}
We consider an undirected graph $\bbG = (\ccalV, \ccalE, \ccalW)$ with $\ccalV, | \ccalV | = n$, representing the set of nodes and $\ccalE \subseteq \ccalV\times \ccalV$ the set of edges. The weight function $\ccalW: \ccalE \rightarrow \reals$ assigns weights to the edges. Graph signals $\bbx \in \reals^n$ are defined as data supported on the graph, with the $i$-th element $x_i$ representing the datum at node $i$ \cite{ortega2018graph,shuman2013emerging}. 

The so-called graph shift operator (GSO) $\bbS\in \reals^{n\times n}$ is typically selected to be the Laplacian matrix, the normalized Laplacian or the adjacency matrix, which encodes the graph sparsity pattern. The application of the GSO to a graph signal leads to a diffusion of the data over the edges. Specifically, for each node $i$, the GSO diffuses the signals on the neighboring nodes $x_j$ scaled with the proximity measure $[\bbS]_{ij}$ to this node $i$. I.e. $[\bbS\bbx]_i=\sum_{j,(i,j)\in\ccalE}[\bbS]_{ij}x_j$. Based on this notion of diffusion, we can then define graph convolutional filters as a shift-and-sum operation
\begin{equation}
    \label{eqn:graph_convolution}
\bbh_\bbG(\bbS) \bbx = \sum_{k=0}^{K-1} h_k \bbS^k \bbx,
\end{equation}
where $\{h_k\}_{k=0}^{K-1}$ are the filter coefficients \cite{gama2019convolutional,ortega2018graph}. 

For an undirected graph $\bbG$, the GSO is symmetric and thus can be diagonalized as $\bbS= \bbV \bm{\Lambda} \bbV^H$, with $\bm{\Lambda}$ containing the graph eigenvalues and $\bbV$ the graph eigenvectors. By replacing the GSO by its eigendecomposition, we can thus write the spectral representation of the graph convolution as 
\begin{equation}
    \label{eqn:graph_convolution_spectral}
    \bbV^H \bbh_\bbG(\bbS) \bbx =  \sum_{k=1}^{K-1} h_k \bm\Lambda^k \bbV^H \bbx = h(\bm\Lambda)\bbV^H \bbx.
\end{equation}
The frequency response of graph convolution can therefore be written as $h(\lambda)= \sum_{k=0}^{K-1} h_k \lambda^k$, which only depends on the coefficients $h_k$  and the eigenvalues of $\bbS$.


\subsection{Manifold signal processing: manifold filters}
We consider a compact smooth differentiable $d$-dimensional submanifold $\ccalM$ embedded in some Euclidean space $\reals^\mathsf{N}$, and a measure $\mu$ over the manifold $\ccalM$. Manifold signals are defined as scalar functions $f:\ccalM\rightarrow \reals$ supported on $\ccalM$. The inner product of functions $f, g \in L^2(\ccalM)$ can be defined as 
\begin{equation}\label{eqn:innerproduct}
    \langle f,g \rangle_{L^2(\ccalM)}=\int_\ccalM f(x)g(x) \text{d}\mu(x) 
\end{equation}
with $\text{d}\mu(x)$ representing the volume element with respect to measure $\mu$.  Therefore, the energy of manifold signal $f$ is naturally given as 
$\|f\|^2_{L^2(\ccalM)}={\langle f,f \rangle_{L^2(\ccalM)}}$. We consider manifold signals with finite energy, i.e., $f \in L^2(\ccalM)$. 

The manifold is locally Euclidean and the local homeomorphic Euclidean space around each point $x\in\ccalM$ is defined as the tangent space $T_x\ccalM$. The disjoint union of all tangent spaces over $\ccalM$ is denoted $T\ccalM$. The \emph{intrinsic gradient} $\nabla: L^2(\ccalM)\rightarrow L^2(T\ccalM)$ is the differentiation operator and maps scalar functions to tangent vector functions over $\ccalM$ \cite{bronstein2017geometric}. The adjoint of $\nabla$ is the \emph{intrinsic divergence}, which is defined as $\text{div}: L^2(T\ccalM)\rightarrow L^2(\ccalM)$. The Laplace-Beltrami (LB) operator $\ccalL: L^2(\ccalM) \to L^2(\ccalM)$ can be defined as the intrinsic divergence of the intrinsic gradient \cite{rosenberg1997laplacian}. More formally,
\begin{equation}\label{eqn:Laplacian}
    \ccalL f=-\text{div}\circ \nabla f.
\end{equation}
The LB operator measures the difference between the function value at a point and the average function value over the point's neighborhood \cite{bronstein2017geometric}, similar to graph Laplacian matrix.

The manifold convolution of $\tdh:\reals^+ \to \reals$ and the manifold signal $f \in L^2(\ccalM)$ is called a manifold filter with impulse response  $\tdh$, and is defined by leveraging the heat diffusion dynamics \cite{wang2022convolutional}. Explicitly, this filter, denoted $\bbh$, is given by
\begin{align} \label{eqn:convolution-conti}
   g(x) = (\bbh f)(x) := \int_0^\infty \tdh(t)e^{-t\ccalL}f(x)\text{d}t =  \bbh(\ccalL)f(x) \text{.}
\end{align}
This is a spatial manifold convolution operator because it is parametrized by the operator $\ccalL$ and directly operates on points $x\in \ccalM$. 

The LB operator $\ccalL$ is self-adjoint and  positive-semidefinite. These facts, together with the compactness of $\ccalM$, imply that $\ccalL$ has a real, positive and discrete eigenvalue spectrum $\{\lambda_i\}_{i=1}^\infty$ such that $\ccalL \bm\phi_i =\lambda_i \bm\phi_i$, where $\bm\phi_i$ is the eigenfunction associated with eigenvalue $\lambda_i$. The eigenvalues are ordered in increasing order as $0<\lambda_1\leq \lambda_2\leq \lambda_3\leq \hdots$. The eigenfunctions are orthonormal and thus form an eigenbasis of $L^2(\ccalM)$ in the intrinsic sense.

Denoting $[\hat{f}]_i=\langle f, \bm\phi_i\rangle_{L^2(\ccalM)}=\int_{\ccalM} f(x)\bm\phi_i(x) \text{d}\mu(x)$ and $e^{-t\ccalL}\bm\phi_i = e^{-t\lambda_i}\bm\phi_i$, the manifold convolution can be represented in the manifold frequency domain as
\begin{align}
    [\hat{g}]_i = \int_0^\infty \tdh(t) e^{-t\lambda_i}  \text{d} t [\hat{f}]_i\text{.}
\end{align}
Hence, the frequency response of the filter $\bbh(\ccalL)$ is given by $\hat{h}(\lambda)=\int_0^\infty \tdh(t) e^{- t \lambda  }\text{d}t$, which depends only on the function $\tilde{h}$ and on the eigenvalues of the LB operator. 
Summing over all $i$, we thus get
\begin{equation}\label{eqn:spectral-filter}
    g =\bbh(\ccalL)f= \sum_{i=1}^{\infty} \hat{h}(\lambda_i)[\hat{f}]_i \bm\phi_i.
\end{equation}



\section{Convergence of filtering on sampled manifolds} \label{sec:convergence}

Our goal is to derive a non-asymptotic approximation bound for manifold filtering on sampled manifolds, providing theoretical guarantees when the sampled manifold is modeled as a graph and showing that graph filtering can be used to approximate manifold filtering. We start by constructing the graph corresponding to the sampled manifold.

\subsection{Sampled manifolds as graphs}
Let $X=\{x_1, x_2, \hdots x_n\}$ be a set of $n$ points sampled uniformly and independently at random from the $d$-dimensional manifold $\ccalM$ according to measure $\mu$. The empirical measure associated with $\text{d}\mu$ is $p_n=\frac{1}{n}\sum_{i=1}^n \delta_{x_i}$, where $\delta_{x_i}$ is the Dirac measure supported on $x_i$. Similar to the definition of inner product in $L^2(\ccalM)$ \eqref{eqn:innerproduct}, the inner product in $L^2(\bbG_n)$ is defined as
 \begin{equation}
     \langle u, v\rangle_{L^2(\bbG_n)}=\int u(x)v(x)\text{d}p_n=\frac{1}{n}\sum_{i=1}^n u(x_i)v(x_i).
 \end{equation}
 The norm in $L^2(\bbG_n)$ with measure $p_n$ is therefore $$\|u\|^2_{L^2(\bbG_n)} = \langle u, u \rangle_{L^2(\bbG_n)}$$ for $u,v \in L^2(\ccalM)$. For signals $\bbu,\bbv \in L^2(\bbG_n)$, the inner product is $$\langle \bbu,\bbv \rangle_{L^2(\bbG_n)} = \frac{1}{n}\sum_{i=1}^n [\bbu]_i[\bbv]_i.$$
 
 From $X$, we construct an undirected graph $\bbG_n$ by seeing the sampled points as nodes. The weights of the edges connecting nodes $x_i$ and $x_j$ are determined by
\begin{equation}\label{eqn:weight}
w_{ij}=\frac{1}{n}\frac{1}{\epsilon(4\pi \epsilon)^{d/2}}\exp\left(-\frac{\|x_i-x_j\|^2}{4\epsilon}\right),
\end{equation}
where $\|x_i-x_j\|$ stands for the Euclidean distance between points $x_i$ and $x_j$ and $\epsilon$ is a parameter related to the Gaussian kernel \cite{belkin2008towards}. The adjacency matrix $\bbA_n \in \reals^{n\times n}$ can be written as $[\bbA_n]_{ij}=w_{ij}$ for $1 \leq i,j\leq n$. Thus, the graph Laplacian associated with this  graph can be written as $$\bbL^\epsilon_n = \mbox{diag}(\bbA_n \boldsymbol{1})-\bbA_n$$
with $[\bbL^\epsilon_n]_{ij} = - w_{ij}$ if $i \neq j$ and $[\bbL^\epsilon_n]_{ii} = \sum_{k=1}^n w_{ik}$ \cite{merris1995survey}. 

To discretize manifold signals $f\in L^2(\ccalM)$ on the sampled manifold, we define a uniform sampling operator $\bbP_n: L^2(\ccalM)\rightarrow L^2(\bbG_n)$ yielding signals $\bbf\in L^2(\bbG_n)$ defined explicitly as
\begin{equation}
\label{eqn:sampling}
    \bbf = \bbP_n f\text{ with }\bbf(x_i) = f(x_i), \quad  x_i \in X.
\end{equation}
This indicates that $\bbf$ shares the function values of the manifold signal $f$ at the sample points $X$.
Seeing the graph Laplacian as an operator acting on $\bbf: X\to \reals$, we can write the diffusion operation on each point $x$ explicitly as
\begin{equation}
\label{eqn:graph_laplacian}
    \bbL^\epsilon_n \bbf(x_i) = \sum_{j=1}^n w_{ij}\left( \bbf(x_i) -\bbf(x_j) \right) , \quad  i = 1,2,\hdots n.
\end{equation}
It can be seen that this formulation can extend to a continuous manifold signal $f$, which we denote as 
\begin{equation}
\label{eqn:discrete_laplacian}
    \bbL^\epsilon_n f(x) = \frac{1}{n}\frac{1}{\epsilon(4\pi \epsilon)^{d/2}} \sum_{j=1}^n \left( f(x) -f(x_j) \right) e^{-\frac{\|x-x_j\|^2}{4\epsilon}}, \; x\in \ccalM.
\end{equation}
If we further extend the set of sample points from $X$ to all the points on the manifold, we finally get the functional approximation of the graph Laplacian,
\begin{equation}
\label{eqn:functional_laplacian}
    \bbL^\epsilon f(x) = \frac{1}{n}\frac{1}{\epsilon(4\pi \epsilon)^{d/2}} \int_\ccalM \left( f(x) -f(y) \right) e^{-\frac{\|x-y \|^2}{4\epsilon}} \text{d}\mu(y),
\end{equation}
which operates on all $x\in\ccalM$.

\subsection{Graph Laplacian approximation of the LB operator}
We first demonstrate the consistency of the discrete graph Laplacians defined in \eqref{eqn:discrete_laplacian} when operating on the eigenfunctions $\bm\phi_i\in L^2(\ccalM)$ of the LB operator $\ccalL$. The quality of the approximation given by the discrete Laplace operator on the sampled manifold is quantified by the following non-asymptotic result.
\begin{theorem}\label{thm:operator-diff}
Let $\ccalM\subset \reals^\mathsf{N}$ be a compact smooth differentiable $d$-dimensional manifold with LB operator $\ccalL$, whose spectrum is given by $\{\lambda_i,\bm\phi_i\}_{i=1}^\infty$, and assume $\bm\phi_i \in C(\ccalM)$. Let $\bbG_n$ be the discrete graph connecting points $\{x_1, x_2,\hdots,x_n\}$ sampled uniformly and independently at random from $\ccalM$, with edge weights as in \eqref{eqn:weight} and $\epsilon=\epsilon(n)> n^{-1/(d+4)}$. Then, it holds 
\begin{equation}
    | \bbL_n^\epsilon \bm\phi_i(x)- \ccalL \bm\phi_i(x)|\leq \left(C_1 \sqrt{ \frac{ \ln(1/\delta)}{2n} }+C_2\sqrt{\epsilon}\right) \lambda_i^{\frac{d+2}{4}},
\end{equation}
with probability at least $1-\delta$. The constants $C_1$, $C_2$ depend on the volume of the manifold. 
\end{theorem}
\begin{proof}
See Appendix \ref{app:operator}.
\end{proof}
This theorem provides a point-wise upper bound for the error incurred when using the discrete graph Laplacian to approximate the LB operator operating on the eigenfunctions of $\ccalL$. We can see that this bound is not only related to the number of sampling points $n$, but also grows with the corresponding eigenvalue $\lambda_i$. This makes sense considering that higher eigenvalues correspond to eigenfunctions with high oscillation \cite{shi2010gradient}.

Based on this approximation result, we can further derive a non-asymptotic approximation bound relating the spectra of the graph Laplacian and the LB operator using the Davis-Khan theorem \cite{yu2015useful}. This result, stated in Theorem \ref{thm:converge-spectrum}, will allows us to analyze the approximation of a manifold filter by a graph filter through their spectral representations.

\begin{theorem}\label{thm:converge-spectrum}
Let $\ccalM\subset \reals^\mathsf{N}$ be a compact smooth differentiable $d$-dimensional manifold with LB operator $\ccalL$ whose spectrum is given by $\{\lambda_i,\bm\phi_i\}_{i=1}^\infty$. Let $\bbL_n^\epsilon$ be a discrete graph Laplacian defined as in \eqref{eqn:discrete_laplacian} with $\{x_1, x_2,\hdots,x_n\}$ sampled uniformly and independently at random from $\ccalM$ whose spectrum is given by $\{\lambda_{i,n}^\epsilon,\bm\phi_{i,n}^\epsilon\}_{i=1}^n$.
Fix $K\in \mathbb{N}$ and asumme that $n$ is sufficiently large so that $\epsilon=\epsilon(n)> n^{-1/(d+4)}$. Then, with probability at least $1-2e^{-n}$, we have 
\begin{equation}
    |\lambda_i-\lambda_{i,n}^\epsilon|\leq \Omega_1 \sqrt{\epsilon}, \quad 
    \|a_i\bm\phi_{i,n}^\epsilon-\bm\phi_i\|\leq \Omega_2 \sqrt{\epsilon},
\end{equation}
with $a_i=\{-1,1\}$ for all $i<K$. The constants $\Omega_1$, $\Omega_2$ depend on $\lambda_K$, the eigengap of $\ccalL$ i.e., $\theta=\min_{1\leq i\leq K}\{\lambda_i-\lambda_{i-1},\lambda_{i+1}-\lambda_{i}\}$, $d$ and the volume of $\ccalM$.
\end{theorem}
\begin{proof}
See Appendix \ref{app:spectrum}.
\end{proof}
\subsection{Graph filtering approximation of manifold filtering}

Equipped with theoretical approximation results for both the eigenvalues and eigenfunctions of the Laplacian operators, we can now prove that graph filters can approximate manifold filters well in the spectral domain. We first show that we can generalize the definition of manifold convolutional filters to sampled manifolds using the same impulse response $\tilde{h}$. 

From the definition in \eqref{eqn:convolution-conti}, recall that the manifold filter is parametric on the LB operator when the impulse response $\tilde{h}(t)$ is fixed. By replacing the LB operator with the discrete graph Laplacian defined in \eqref{eqn:discrete_laplacian}, we can then obtain a graph convolutional filter
\begin{equation}
    \bbg = \int_0^\infty \tilde{h}(t) e^{-t\bbL_n^\epsilon} \bbf \text{d}t =\bbh(\bbL_n^\epsilon)\bbf, \quad \bbg, \bbf \in \reals^n.
\end{equation}
This can be seen as a \textit{continuous-time} graph filtering process, different from the \textit{discrete-time} graph filtering process $\bbh_\bbG$ in \eqref{eqn:graph_convolution}, where the exponential term $e^{-\bbL_n^\epsilon}$ can be seen as the graph shift operator.
The frequency representation of this graph filter can be written as
\begin{equation}
   \bbg  = \sum_{i=1}^n \hat{h}(\lambda_{i,n}^\epsilon)\langle\bbf, \bm\phi_{i,n}^\epsilon \rangle_{L^2(\bbG_n)}\bm\phi_{i,n}^\epsilon,
\end{equation}
which reveals the full dependency on the eigendecomposition of $\bbL_n^\epsilon$. Thus, we can relate graph filtering and manifold filtering using the spectral approximation results for the graph Laplacian and LB operator presented in Theorem \ref{thm:converge-spectrum}.

Unlike the finite spectrum of the graph Laplacian, the LB operator $\ccalL$ possesses an infinite spectrum. We consider Weyl's law \cite{arendt2009weyl} when analyzing the spectral properties of $\ccalL$. This classical result states that the eigenvalues $\lambda_i$ of $\ccalL$ grow in the order of $i^{2/d}$. Therefore, the difference between neighboring eigenvalues becomes quite small in the high frequency spectrum, i.e., large eigenvalues accumulate. Leveraging this fact, we can use a partition strategy to separate the spectrum into finite groups as in Definition \ref{def:alpha-spectrum}.


\begin{definition}\cite{wang2021stability} ($\alpha$-separated spectrum)\label{def:alpha-spectrum}
The $\alpha$-separated spectrum of LB operator $\ccalL$ is defined as the partition $\Lambda_1(\alpha) \cup \ldots\cup \Lambda_N(\alpha)$ such that $\lambda_i \in \Lambda_k(\alpha)$ and $\lambda_j \in \Lambda_l(\alpha)$, $k \neq l$, satisfying
$|\lambda_i - \lambda_j| > \alpha$.
\end{definition}

The following defined $\alpha$-FDT filter allows obtaining the $\alpha$-separated spectrum as in Definition \ref{def:alpha-spectrum}.

\begin{definition}\cite{wang2021stability} ($\alpha$-FDT filter)\label{def:alpha-filter}
The $\alpha$-frequency difference threshold ($\alpha$-FDT) filter is defined as a filter $\bbh(\ccalL)$ whose frequency response satisfies
\begin{equation} \label{eq:fdt-filter}
    |\hhath(\lambda_i)-\hhath(\lambda_j)|\leq \gamma_k \mbox{ for all } \lambda_i, \lambda_j \in \Lambda_k(\alpha) 
\end{equation}
with $\gamma_k\leq \gamma$ for $k=1, \ldots,N$.
\end{definition}

To derive an approximation bound, we will also need the manifold filter to have a Lipschitz frequency response as in Definition \ref{def:lipschitz}.

\begin{definition}[Lipschitz filter] \label{def:lipschitz}
A filter is $A_h$-Lispchitz if its frequency response is Lipschitz continuous with Lipschitz constant $A_h$,
\begin{equation}
    |\hhath(a)-\hhath(b)| \leq A_h |a-b|\text{ for all } a,b \in (0,\infty)\text{.}
\end{equation}
\end{definition}

Equipped with the above requirements for the filter frequency response, we can finally establish the upper bound on the manifold filter approximation error on the sampled manifold.

\begin{theorem}
\label{thm:converge-MNN} 
Let $X=\{x_1, x_2,...x_n\}$ be a set of $n$ points sampled by an operator $\bbP_n$ \eqref{eqn:sampling} from a $d$-dimensional manifold $\ccalM \subset \reals^N$. Let $\bbG_n$ be a discrete graph approximation of $\ccalM$ constructed from $X$ with weight values as in \eqref{eqn:weight} with $\epsilon = \epsilon(n)\geq  n^{-1/(d+4)}$. Let $\bbh(\cdot)$ be the convolutional filter parameterized by the LB operator $\ccalL$ of the manifold $\ccalM$ \eqref{eqn:convolution-conti} or by the discrete graph Laplacian operator $\bbL_n^\epsilon$ of the graph $\bbG_n$. Under the assumption that the frequency response of filter $\bbh$ is Lipschitz continuous and $\alpha$-FDT with $\alpha^2 \gg \epsilon$ and $\gamma = \Omega_2'\sqrt{\epsilon}/\alpha$, with probability at least $1-2n^{-2}$ it holds that
\begin{align}
&\nonumber  \nonumber \|\bbh(\bbL_n^\epsilon)\bbP_n f - \bbP_n\bbh( \ccalL)f\|_{L^2(\bbG_n)}\\&\qquad \qquad \qquad 
\leq
 \left(\frac{ N\Omega_2'}{\alpha} +A_h \Omega_1 \right)\sqrt{\epsilon}+C_{gc}\sqrt{\frac{\log n}{{n}}}
\end{align}
where $N$ is the partition size of $\alpha$-FDT filter and $C_{gc}$ is related with $d$ and the volume of $\ccalM$.
\end{theorem}
\begin{proof}
See Appendix \ref{app:nn}.
\end{proof}

From this theorem, we see that if we take $\epsilon = n^{-1/(d+4)}$, the difference between filtering on the manifold and the sampled manifold is in the order of $O(n^{-1/(2d+8)})$. Thus, with a large enough number of sampling points, this approximation provides a good accuracy with high probability. This provides a formal theoretical guarantee for the ability of graph filters to approximate manifold filters on sampled manifolds, enabling approximate linear information processing on continuous non-Euclidean domains.


\section{Numerical experiments} \label{sec:sims}




We consider a problem of automatic navigation control of an agent \cite{cervino2022learning}. We intend to control an agent to find a path to the goal without colliding to obstacles. We sample grid points over the free space which means avoiding the obstacle areas to construct a graph structure.  We generate four trajectories with Dijkstra's shortest path algorithm leading a starting point to the goal. Each point along the trajectory is labeled with the direction of the velocity. The adjacency matrix can be calculated based on the weight value defined in \eqref{eqn:weight} with the Euclidean distance between each two points. We note that if the direct path between two points goes through the obstacles, we set the distance between them as infinity, i.e. the weight value as zero, so as to better capture the geometric structure. The input graph signal is the position of the nodes and the output of the graph filter is the direction of this node leading to the goal. We construct a one-layer and two-layer graph filters with $5$ filter taps in each layer. We train all the architectures for $30,000$ epochs with SGD optimizer with the learning rate set as $0.0002$. The performance of the learned graph filter is testified by randomly generating $100$ starting points and computing the trajectories. If the trajectory can reach the goal without colliding into the obstacles, the trajectory is marked as ``success".

\begin{figure}
     \begin{subfigure}[b]{0.2\textwidth}
         \centering
         \includegraphics[width=1.3\textwidth]{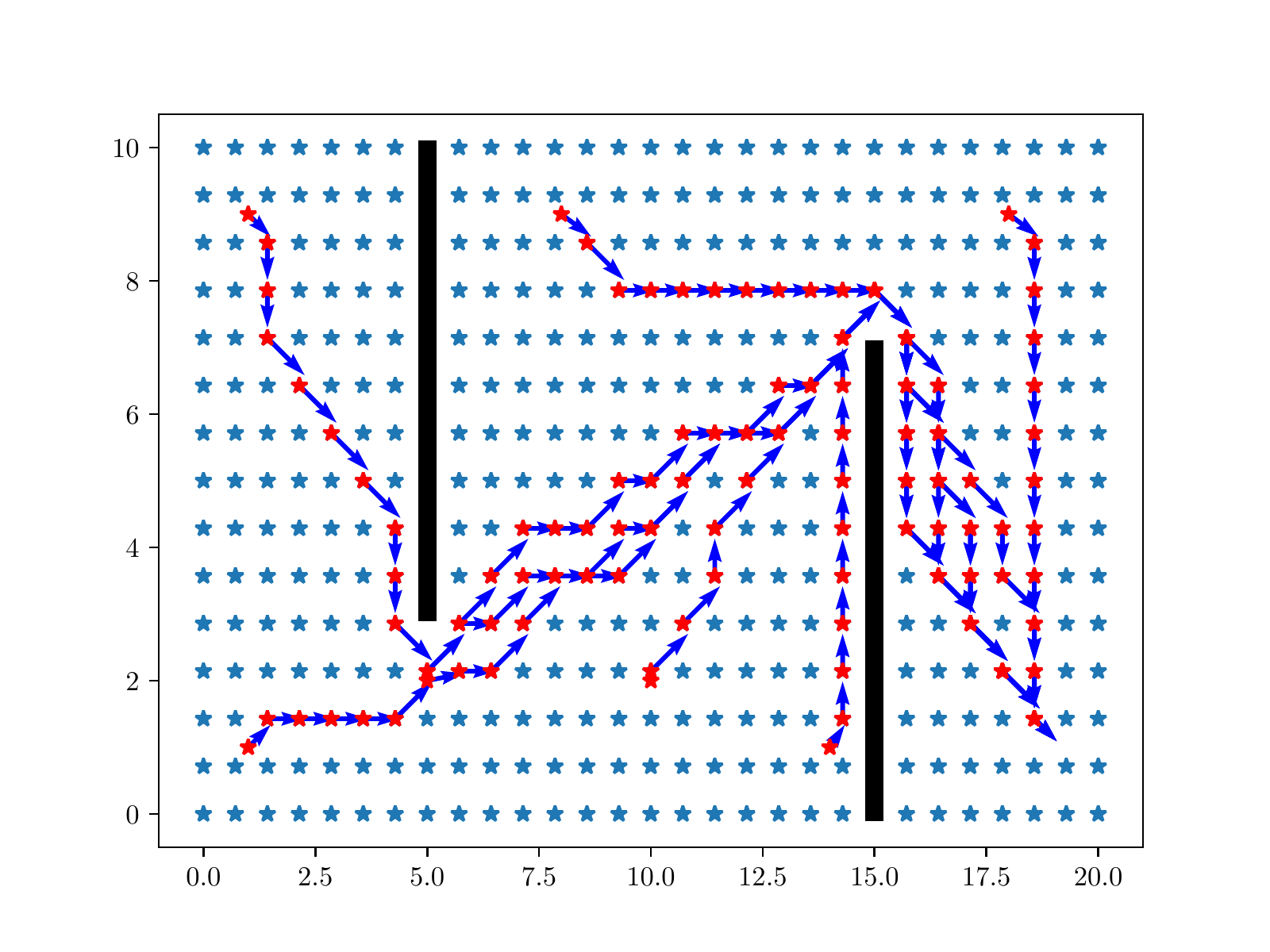}
         \caption{Training dataset}
         \label{fig:train400}
     \end{subfigure}
    \hspace{1.5em}
     \begin{subfigure}[b]{0.2\textwidth}
         \centering
         \includegraphics[width=1.3\textwidth]{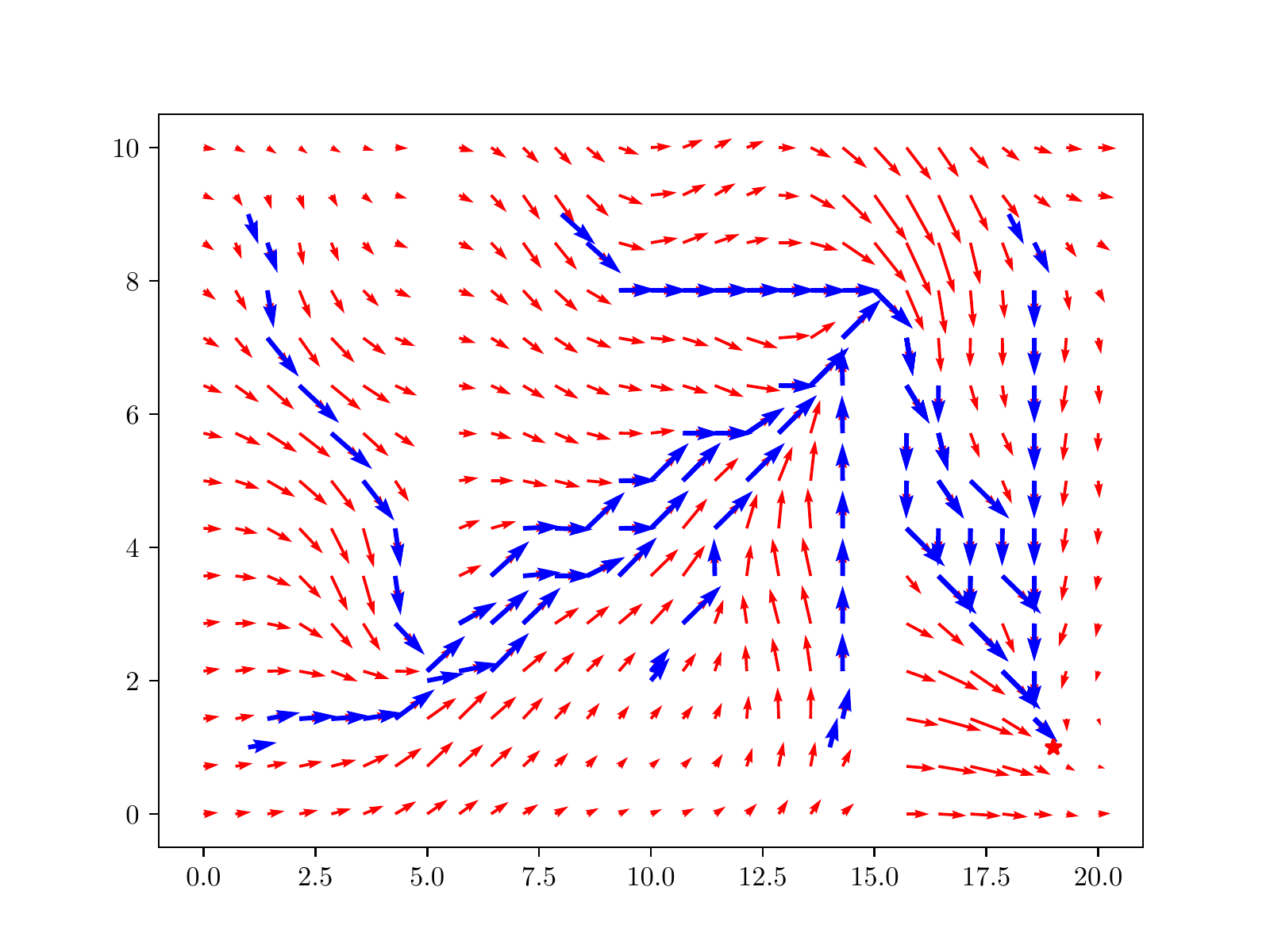}
         \caption{Learned direction}
         \label{fig:learned400}
     \end{subfigure}
   
        \caption{Graph filtering performance with $n=413$}
        \label{fig:400nodes}
\end{figure}

\begin{figure}
     \begin{subfigure}[b]{0.2\textwidth}
         \centering
         \includegraphics[width=1.3\textwidth]{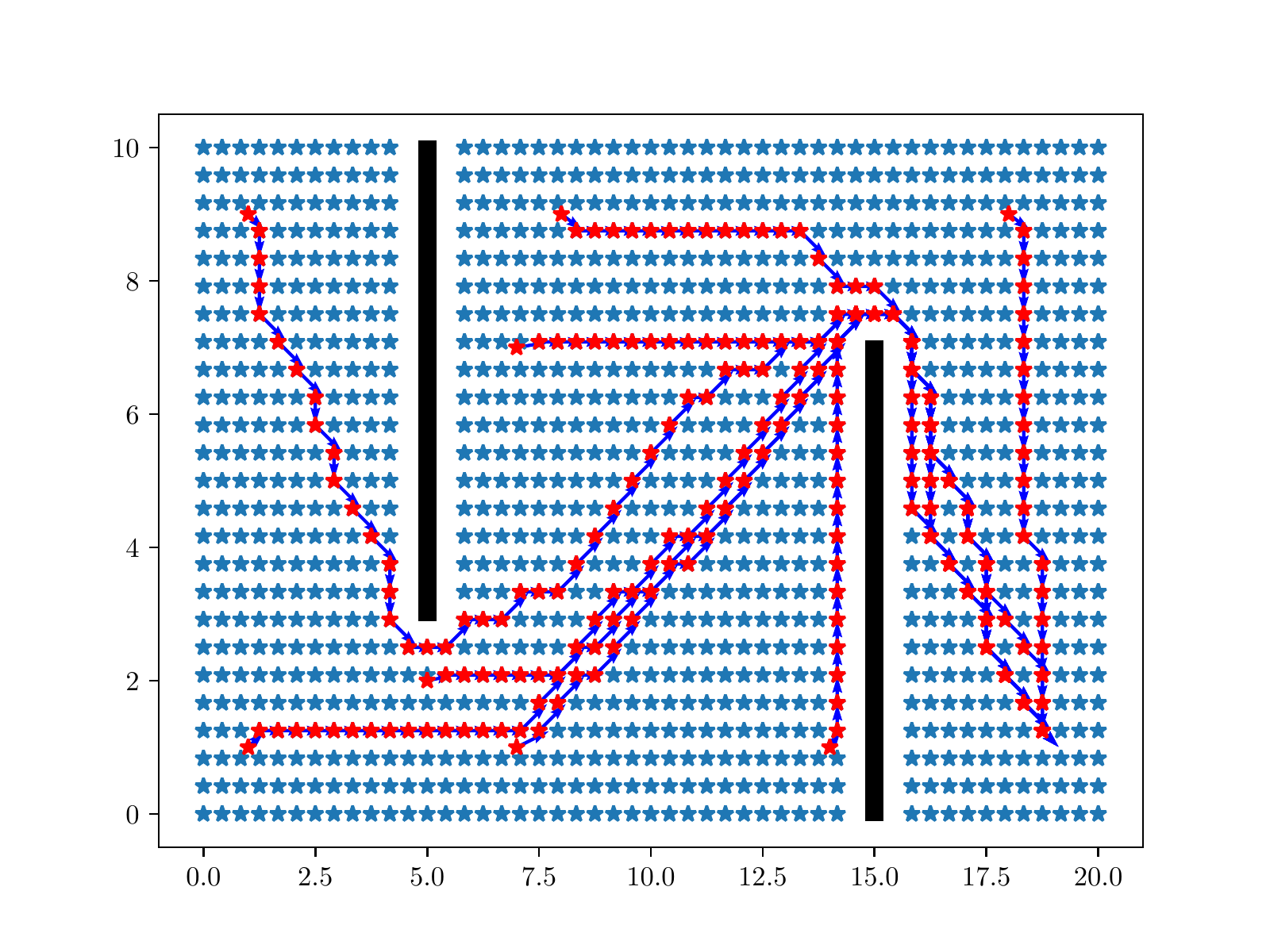}
         \caption{Training dataset}
         \label{fig:train1000}
     \end{subfigure}
    \hspace{1.5em}
     \begin{subfigure}[b]{0.2\textwidth}
         \centering
         \includegraphics[width=1.3\textwidth]{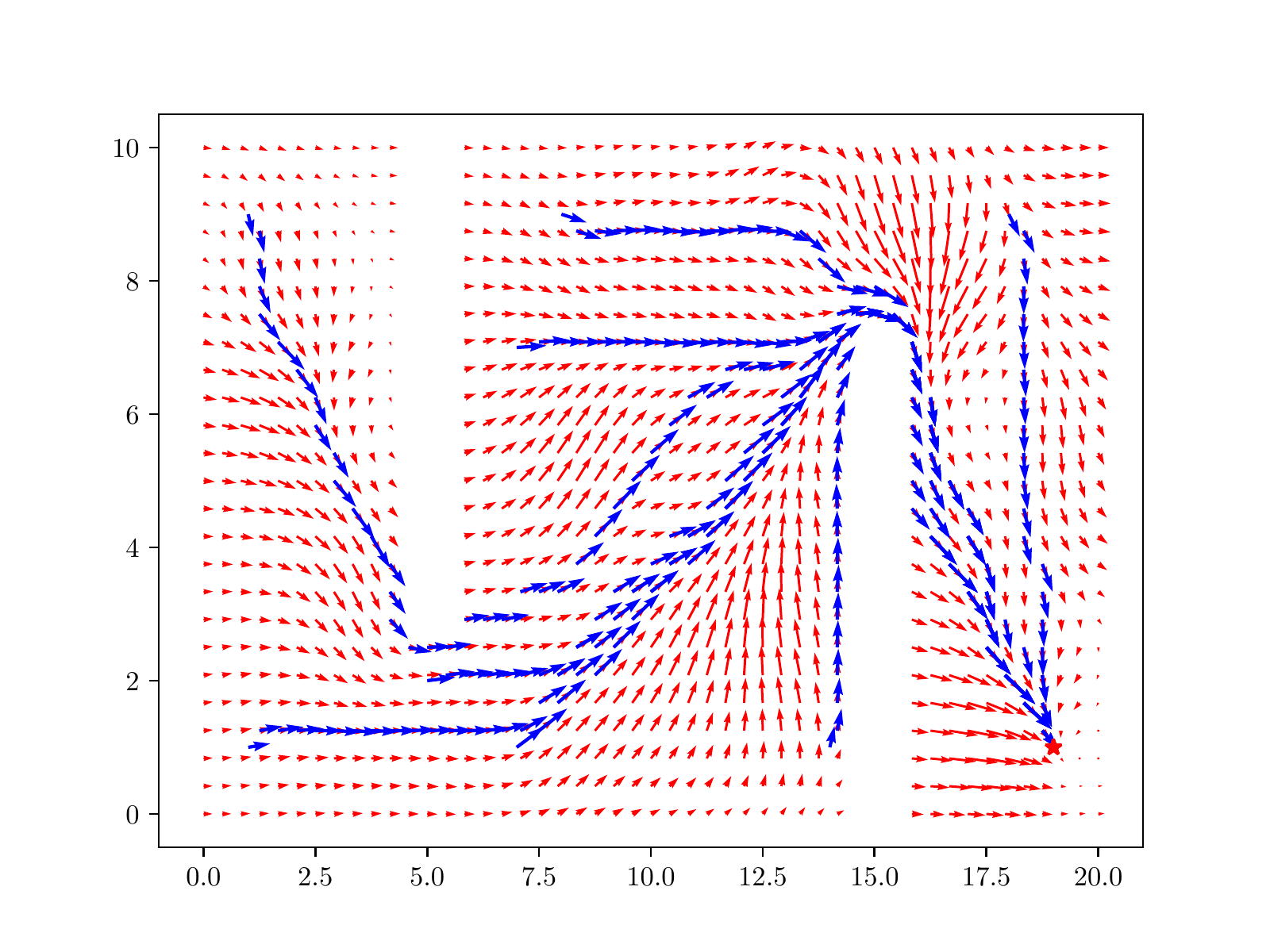}
         \caption{Learned direction}
         \label{fig:learned1000}
     \end{subfigure}
   
        \caption{Graph filtering performance with $n=1117$}
        \label{fig:1000nodes}
\end{figure}
We first sample $n$ points uniformly over the free space avoiding the black obstacles as shown in Fig. \ref{fig:train400} ($n=413$) and Fig. \ref{fig:train1000} ($n=1117$). In Fig. \ref{fig:train400} and Fig. \ref{fig:train1000}, the red stars depicting labeled points as training dataset and the blue arrows attached to the red stars point to the optimal directions leading to the goal. The blue stars depict unlabeled points over the free space. We implement a 1-layer or 2-layer graph filters to learn and estimate the directions for all the unlabeled points.   Figure \ref{fig:learned400} and \ref{fig:learned1000} show the learned directions starting from each unlabeled points as the red arrows show using graph filtering. The numbers of successful trajectories within $100$ generated testing are shown in Table \ref{tb:results}. We can see that graph filtering can efficiently learn the potential successful trajectories for unlabeled points based on these sampled manifolds. More sampling points characterize the free space more accurately, which leads to more successful testing trajectories. This is also in accordance with our theoretical result presented in Theorem \ref{thm:converge-MNN}. Moreover, more layers of filters also learn a more accurate prediction.

\vspace{-3mm}
\begin{table}[h!]
\centering
\begin{tabular}{l|c| c} \hline
    & $n=413$  & $n=1117$  \\ \hline
1-layer graph filter		&74  & 75  \\ \hline
2-layer graph filter	& 79  &  84 \\ \hline
\end{tabular}
\caption{Number of successful trajectories. }
\label{tb:results}
\vspace{-4mm}
\end{table}


\section{Conclusions} \label{sec:conclusions}

In this paper, we introduced manifold convolutional filters defined with the exponential Laplace-Beltrami operator to process geometric data. To access the manifold model, we sample uniformly over the manifold and construct a graph model as a sampled manifold. We first prove the non-asymptotic approximation error bound of discrete graph Laplacian to LB operator in both operator and spectral aspects. By transferring the manifold filtering to the sampled manifold, the approximation of the graph filtering to the manifold filtering can be derived. The non-asymptotic error bound decreases with the number of sampling points growing. The approximation of graph filtering is verified empirically on a navigation control problem over a manifold.


\bibliographystyle{IEEEtran}
\bibliography{references}

\appendix
 {\section{Appendix}
 \subsection{Proof of Theorem \ref{thm:operator-diff}}
\label{app:operator}
We decompose the operator difference between the graph Laplacian and the LB operator with an intermediate term $\bbL^\epsilon$, which is the functional approximation defined in \eqref{eqn:functional_laplacian}. We first focus on the operator difference between $\bbL^\epsilon$ and $\ccalL$. From \cite{belkin2008towards}, we can get the bound as
\begin{equation}
    \|\bbL^\epsilon \bm\phi_i- \ccalL \bm\phi_i\|\leq C\sqrt{\epsilon} \|\bm\phi_i\|_{H^{d/2+1}},
\end{equation}
For the Sobolev norm of eigenfunction $\bm\phi_i$, according to \cite[Lemma~4.4]{belkin2006convergence} we have
\begin{equation}
    \|\bm\phi_i\|_{H^{d/2+1}}\leq C \lambda_i^{\frac{d+2}{4}},
\end{equation}
which leads to 
\begin{equation}
\label{eqn:Lepsilon-L}
    \|\bbL^\epsilon \bm\phi_i- \ccalL \bm\phi_i\|\leq C_1 \sqrt{\epsilon} \lambda_i^{\frac{d+2}{4}}.
\end{equation}
For the operator difference between $\bbL_n^\epsilon$ and $\bbL^\epsilon$ with Hoeffding's inequality as 
\begin{align}
    \mathbb{P}\left( |\bbL_n^\epsilon \bm\phi_i(x) - \bbL^\epsilon \bm\phi_i(x)|> \epsilon_1\right) \leq \exp\left( - \frac{2n\epsilon_1^2}{\|\bm\phi_i\|_{H^{d/2+1}}^2} \right).
\end{align}
Therefore, we can claim that with probability at least $1-\delta$, we have
\begin{align}
\label{eqn:Lnepsilon-Lepsilon}
    |\bbL_n^\epsilon \bm\phi_i(x) - \bbL^\epsilon \bm\phi_i(x)|\leq \sqrt {\frac{\ln 1/\delta}{2n}} \|\bm\phi_i\|_{H^{d/2+1}}.
\end{align}
Combining \eqref{eqn:Lepsilon-L} and \eqref{eqn:Lnepsilon-Lepsilon} with triangle inequality, we can get the conclusion in Theorem \ref{thm:operator-diff}.

\subsection{Proof of Theorem \ref{thm:converge-spectrum}}
We first import two lemmas to help prove the spectral properties.
\label{app:spectrum}
\begin{lemma}\label{lem:conv_eigenfunction}
Let $\bbA, \bbB$ be self-adjoint operators with $\{\lambda_i(\bbA), \bbu_i\}_{i=1}^\infty$ and $\{\lambda_i(\bbB), \bbw_i\}_{i=1}^\infty$ as the corresponding spectrum. Let $Pr_{\bbw_i}$ be the orthogonal projection operation onto the subspace generated by $\bbw_i$. Then we have
\begin{equation}
    \|a_i\bbu_i-\bbw_i\|\leq 2\|\bbu_i-Pr_{\bbw_i}\bbu_i\|\leq \frac{2\|\bbB\bbu_i-\bbA\bbu_i\|}{\min_{j\neq i} |\lambda_j(\bbB)-\lambda_i(\bbA)|}.
\end{equation}
\end{lemma}
\begin{proof}
The first inequality is directly from \cite[Proposition~18]{von2008consistency}. Let $Pr^\perp_{\bbw_i}$ be the orthogonal projection onto the complement of the subspace generated by $\bbw_i$. Then we have
\begin{align}
    \|\bbu_i - Pr_{\bbw_i} \bbu_i\|=\|Pr_{\bbw_i}^\perp \bbu_i\|=\Big\|\sum_{j\neq i}\langle \bbu_i, \bbw_j \rangle \bbw_j\Big\|.
\end{align}
Therefore, we have
\begin{align}
  \nonumber  &\|Pr_{\bbw_i}^\perp \bbB\bbu_i - Pr_{\bbw_i}^\perp \bbA\bbu_i\| \\
    & = \Big\| \sum_{j\neq i}\langle \bbB\bbu_i,\bbw_j \rangle \bbw_j -\sum_{j\neq i}\langle \bbA\bbu_i,\bbw_j \rangle\bbw_j \Big\|\\
    &=\Big\| \langle \bbu_i, \bbB\bbw_j \rangle\bbw_j -\sum_{j\neq i} \lambda_i(\bbA)\langle\bbu_i, \bbw_j \rangle \bbw_j\Big\|\\
    &=\Big\| \sum_{j\neq i}(\lambda_i(\bbB)-\lambda_i(\bbA)) \langle \bbu_i,\bbw_j \rangle\bbw_j\Big\|\\
    &\geq \min_{j\neq i} | \lambda_i(\bbB)-\lambda_i(\bbA) | \|\sum_{j\neq i}\langle \bbu_i,\bbw_j \rangle\bbw_j \| \\
    &= \min_{j\neq i} | \lambda_i(\bbB)-\lambda_i(\bbA) | \| \bbu_i -Pr_{\bbw_i}\bbu_i \|,
\end{align}
together with $\|\bbB\bbu_i - \bbA\bbu_i \| \geq  \|Pr_{\bbw_i}^\perp \bbB\bbu_i - Pr_{\bbw_i}^\perp \bbA\bbu_i\|$. We can conclude the proof.
\end{proof}
The following lemma is adapted from \cite[Lemma~5c]{dunson2021spectral}
\begin{lemma}\label{lem:conv_eigenvalue}
Let $\bbA, \bbB$ be self-adjoint operators with $\{\lambda_i(\bbA), \bbu_i\}_{i=1}^\infty$ and $\{\lambda_i(\bbB), \bbw_i\}_{i=1}^\infty$ as the corresponding spectrum. Then we have
\begin{equation}
    |\lambda_i(\bbA)-\lambda_i(\bbB)|=\frac{\langle(\bbA-\bbB)\bbu_i,\bbw_i\rangle}{|\langle \bbu_i,\bbw_i \rangle|}\leq \frac{\|(\bbA-\bbB)\bbu_i\|}{|\langle \bbu_i,\bbw_i \rangle|}
\end{equation}
\end{lemma}

With the above lemmas and our proposed Theorem \ref{thm:operator-diff}, which includes the operator difference, we can prove Theorem \ref{thm:converge-spectrum}. We first fix some $K\in\mathbb{N}$, wich provides an upper bound for $\lambda_i\leq \lambda_K$ for all $1\leq i\leq K$. By taking the probability $1-n^{-2}$ and $\epsilon = n^{-1/(d+4)}$, we can conclude that the operator difference in Theorem \ref{thm:operator-diff} can be bounded with order $O(\sqrt{\epsilon})$, with the constant scaling with $\lambda_K^{\frac{d+2}{4}}$. Combine with Lemma \ref{lem:conv_eigenfunction} and $theta =\min_{1\leq j\neq i \leq K} |\lambda_j-\lambda_{i,n}^\epsilon|$, we can get 
\begin{equation}
    \|a_i\bm\phi_{i,n}^\epsilon - \bm\phi_i\|\leq \frac{C_k}{\theta}\sqrt{\epsilon},
\end{equation}
where we denote the constant as $\Omega_1$ to include the effects of $K$, the eigengap and the volume of $\ccalM$.

This upper bound of the eigenfunction difference leads to $|\langle\bbu_i,\bbw_i \rangle |\geq 1-\Omega_1/2\sqrt{\epsilon}\geq 1$. Combining with Lemma \ref{lem:conv_eigenvalue}, the difference of the eigenvalues can also be bounded in the order of $O(\sqrt{\epsilon})$.

\subsection{Proof of Theorem \ref{thm:converge-MNN}}
\label{app:nn}
We first write out the filter representation as
 \begin{align}
    &\nonumber \|\bbh(\bbL_n^\epsilon)\bbP_n f - \bbP_n\bbh(\ccalL) f\|\\
    &\leq \left\| \sum_{i=1}^\infty \hat{h}(\lambda_{i,n}^\epsilon) \langle \bbP_nf,\bm\phi_{i,n}^\epsilon \rangle_{\bbG_n}\bm\phi_i^n - \sum_{i=1}^\infty \hat{h}(\lambda_i)\langle f,\bm\phi_i\rangle_{\ccalM} \bbP_n \bm\phi_i  \right\|
 \end{align}
 
 We decompose the $\alpha$-FDT filter function as $\hat{h}(\lambda)=h^{(0)}(\lambda)+\sum_{l\in\ccalK_m}h^{(l)}(\lambda)$ as
 \begin{align}
\label{eqn:h0}& h^{(0)}(\lambda) = \left\{ 
\begin{array}{cc} 
              \hat{h}(\lambda)-\sum\limits_{l\in\ccalK_m}\hat{h}(C_l)  &  \lambda\in[\Lambda_k(\alpha)]_{k\in\ccalK_s} \\
                0& \text{otherwise}  \\
                \end{array} \right. \\
\label{eqn:hl}& h^{(l)}(\lambda) = \left\{ 
\begin{array}{cc} 
                \hat{h}(C_l) &  \lambda\in[\Lambda_k(\alpha)]_{k\in\ccalK_s} \\
                \hat{h}(\lambda) & 
                \lambda\in\Lambda_l(\alpha)\\
                0 &
                \text{otherwise}  \\
                \end{array} \right.             
\end{align}
 
 With the triangle inequality, we start by analyzing the output difference of $h^{(0)}(\lambda)$ as
 \begin{align}
    & \nonumber \left\| \sum_{i=1}^{N_s} {h}^{(0)}(\lambda_i^n) \langle \bbP_nf,\bm\phi_i^n \rangle_{\bbG_n}\bm\phi_i^n - \sum_{i=1}^{N_s} {h}^{(0)}(\lambda_i)\langle f,\bm\phi_i\rangle_{\ccalM} \bbP_n \bm\phi_i  \right\|
     \\ 
     &\nonumber \leq  \left\| \sum_{i=1}^{N_s} \left({h}^{(0)}(\lambda_i^n)- {h}^{(0)}(\lambda_i) \right) \langle \bbP_nf,\bm\phi_i^n \rangle_{\bbG_n}\bm\phi_i^n \right\| \\
     &  +\left\| \sum_{i=1}^{N_s} {h}^{(0)}(\lambda_i)\left( \langle \bbP_n f,\bm\phi_i^n \rangle_{\bbG_n} \bm\phi_i^n - \langle f,\bm\phi_i \rangle_{\ccalM} \bbP_n \bm\phi_i \right)  \right\|.\label{eqn:conv-1}
 \end{align}
 
 The first term in \eqref{eqn:conv-1} can be bounded by leveraging the $A_h$-Lipschitz continuity of the frequency response. From the eigenvalue difference in Theorem \ref{thm:converge-spectrum}, we can claim that for each eigenvalue $\lambda_i \leq \lambda_K$, we have
\begin{gather}
 \label{eqn:eigenvalue}  |\lambda_{i,n}^\epsilon-\lambda_i|\leq \Omega_1\sqrt{\epsilon}.
 \end{gather}
The first term is bounded as 
\begin{align}
   &\nonumber \left\| \sum_{i=1}^{N_s} ({h}^{(0)}(\lambda_i^n) - {h}^{(0)}(\lambda_i)) \langle \bbP_n f,\bm\phi_i^n \rangle_{\bbG_n} \bm\phi_i^n  \right\|\\
   & \leq \sum_{i=1}^{N_s} |{h}^{(0)}(\lambda_i^n)-{h}^{(0)}(\lambda_i)| |\langle \bbP_n f,\bm\phi_i^n \rangle_{\bbG_n}| \|\bm\phi_i^n\|\\
   &\leq \sum_{i=1}^{N_s} A_h |\lambda_i^n-\lambda_i| \|\bbP_n f\| \|\bm\phi_i^n \|^2\leq N_s A_h\Omega_1 \sqrt{\epsilon}.
\end{align}

The second term in \eqref{eqn:conv-1} can be bounded combined with the convergence of eigenfunctions in \eqref{eqn:eigenfunction} as
\begin{align}
  & \nonumber \Bigg\| \sum_{i=1}^{N_s} {h}^{(0)}(\lambda_i)\left( \langle \bbP_nf,\bm\phi_i^n \rangle_{\bbG_n}\bm\phi_i^n - \langle f,\bm\phi_i \rangle_{\ccalM} \bbP_n \bm\phi_i\right)  \Bigg\|\\
   & \leq \nonumber \Bigg\|  \sum_{i=1}^{N_s} {h}^{(0)}(\lambda_i)  \left(\langle \bbP_n f,\bm\phi_i^n\rangle_{\bbG_n}\bm\phi_i^n  - \langle \bbP_nf,\bm\phi_i^n \rangle_{\bbG_n} \bbP_n\bm\phi_i\right)\Bigg\|\\
   &\label{eqn:term1}+ \left\| \sum_{i=1}^{N_s}  {h}^{(0)}(\lambda_i) \left(\langle \bbP_n f,\bm\phi_i^n\rangle_{\bbG_n} \bbP_n\bm\phi_i -\langle f,\bm\phi_i\rangle_\ccalM \bbP_n\bm\phi_i \right) \right\|
\end{align}
From the convergence stated in Theorem \ref{thm:converge-spectrum}, we have
\begin{gather}
 \label{eqn:eigenfunction}    \|a_i \bm\phi_{i,n}^\epsilon-\bm\phi_i\|\leq \Omega_2 \sqrt{\epsilon},
 \end{gather}
 Therefore, the first term in \eqref{eqn:term1} can be bounded as
\begin{align}
& \nonumber \left\|  \sum_{i=1}^{N_s} {h}^{(0)}(\lambda_i) \left(\langle \bbP_n f,\bm\phi_i^n\rangle_{\bbG_n}\bm\phi_i^n  - \langle \bbP_nf,\bm\phi_i^n \rangle_{\ccalM} \bbP_n\bm\phi_i\right)\right\|\\
& \qquad \qquad\leq \sum_{i=1}^{N_s} \|\bbP_n f\|\|\bm\phi_i^n - \bbP_n\bm\phi_i\|\leq N_s \Omega_2\sqrt{\epsilon}.
\end{align}
The last equation comes from the definition of norm in $L^2(\bbG_n)$.
The second term in \eqref{eqn:term1} can be written as
\begin{align}
     & \nonumber \Bigg\| \sum_{i=1}^{N_s}  {h}^{(0)}(\lambda_i^n) (\langle \bbP_n f,\bm\phi_i^n\rangle_{\bbG_n}  \bbP_n\bm\phi_i -\langle f,\bm\phi_i\rangle_\ccalM \bbP_n\bm\phi_i ) \Bigg\| \\
   &\leq \sum_{i=1}^{N_s} |{h}^{(0)}(\lambda_i^n)| \left|\langle \bbP_n f,\bm\phi_i^n\rangle_{\bbG_n}  -\langle f,\bm\phi_i\rangle_\ccalM\right|\|\bbP_n\bm\phi_i\|.
\end{align}
Because $\{x_1, x_2,\cdots,x_n\}$ is a set of uniform sampled points from $\ccalM$, based on Theorem 19 in \cite{von2008consistency} we can claim that
\begin{equation}
   \left|\langle \bbP_n f,\bm\phi_i^n\rangle_{\bbG_n}  -\langle f,\bm\phi_i\rangle_\ccalM\right| = O\left(\sqrt{\frac{\log n}{n}}\right).
\end{equation}
Taking into consider the boundedness of frequency response $|{h}^{(0)}(\lambda)|\leq 1$ and the bounded energy $\|\bbP_n\bm\phi_i\|$. Therefore, we have 
\begin{align}
&\nonumber  \left\| \sum_{i=1}^{N_s} \hat{h}(\lambda_i^n) \left(\langle \bbP_n f,\bm\phi_i^n\rangle_{\bbG_n}  -\langle f,\bm\phi_i\rangle_\ccalM \right)\bbP_n\bm\phi_i  \right\|= O\left(\sqrt{\frac{\log n}{n}}\right).
\end{align}

Combining the above results, we can bound the output difference of $h^{(0)}$. Then we need to analyze the output difference of $h^{(l)}(\lambda)$ and bound this as
\begin{align}
    \nonumber &\left\| \bbP_n \bbh^{(l)}(\ccalL)f -\bbh^{(l)}(\bbL_n)\bbP_n f \right\| 
    \\& \leq \left\| (\hat{h}(C_l)+\delta)\bbP_n f - (\hat{h}(C_l)-\delta)\bbP_nf\right\| \leq 2\delta\|\bbP_nf\|,
\end{align}
where $\bbh^{(l)}(\ccalL)$ and $\bbh^{(l)}(\bbL_n)$ are filters with filter function $h^{(l)}(\lambda)$ on the LB operator $\ccalL$ and graph Laplacian $\bbL_n$ respectively.
Combining the filter functions, we can write
\begin{align}
   \nonumber &\|\bbP_n\bbh(\ccalL)f-\bbh(\bbL_n)\bbP_n f\|\\\nonumber &=
    \Bigg\|\bbP_n\bbh^{(0)}(\ccalL)f +\bbP_n\sum_{l\in\ccalK_m}\bbh^{(l)}(\ccalL)f -\\& \qquad \qquad \qquad \bbh^{(0)}(\bbL_n)\bbP_n f - \sum_{l\in\ccalK_m} \bbh^{(l)}(\bbL_n)\bbP f \Bigg\|\\
    &\nonumber \leq \|\bbP_n \bbh^{(0)}(\ccalL)f-\bbh^{(0)}(\bbL_n)\bbP_n f\|+\\
    &\qquad \qquad \qquad \sum_{l\in\ccalK_m}\|\bbP_n \bbh^{(l)}(\ccalL)f-\bbh^{(l)}(\bbL_n)\bbP_nf\|.
\end{align}

Above all, we can claim that 
\begin{equation}
    \|\bbh(\bbL_n)\bbP_n f - \bbP_n\bbh(\ccalL) f\|\leq \left( N\Omega_2 +A_h \Omega_1\right)\sqrt{\epsilon} + C\sqrt{\frac{\log n}{n}}
\end{equation}

}

\end{document}